\documentclass[12pt,3p]{elsarticle}
\usepackage{blindtext}
\usepackage{fancyhdr}
\usepackage[fleqn]{amsmath}
\usepackage{amsthm,enumitem}
\usepackage{array}
\usepackage{graphicx}
\usepackage{nccmath}
\usepackage{amssymb}
\usepackage{color}
\usepackage{etoolbox}
\usepackage{algorithm}
\usepackage{epstopdf}
\usepackage[noend]{algpseudocode}
\usepackage{setspace}
\usepackage{thmtools}
\usepackage[table]{xcolor}
\usepackage{tabularx}

\usepackage[bottom]{footmisc} 
\usepackage{fancyhdr} 
\newcounter{MYtempeqncnt}
\allowdisplaybreaks
\makeatletter
\def\ps@pprintTitle{%
\let\@oddhead\@empty
\let\@evenhead\@empty
\def\@oddfoot{}%
\let\@evenfoot\@oddfoot}
 \makeatother
 
\begin{document}
\journal{Computers $\&$ Electrical Engineering}
\begin{frontmatter}

\author{Soheil Khavari Moghaddam and S. Mohammad Razavizadeh\\ 
School of Electrical Engineering, Iran University of Science and Technology, Tehran, Iran\\
khavari@elec.iust.ac.ir, smrazavi@iust.ac.ir }
\title{Joint Tilt Angle Adaptation and Beamforming in Multicell Multiuser Cellular Networks}

%\IEEEpeerreviewmaketitle
\begin{abstract}
3D beamforming is a promising approach for interference coordination in cellular networks which brings significant improvements in comparison with conventional 2D beamforming techniques.  This paper investigates the problem of joint beamforming design and tilt angle adaptation of the BS antenna array for maximizing energy efficiency (EE) in downlink of multi-cell multi-user coordinated cellular networks. An iterative algorithm based on fractional programming approach is introduced to solve the resulting non-convex optimization problem. In each iteration, users are clustered based on their elevation angle. Then, optimization of the tilt angle is carried out through a reduced complexity greedy search to find the best tilt angle for a given placement of the users. Numerical results show that the proposed algorithm achieves higher EE compared to the 2D beamforming techniques. 

\end{abstract}


\end{frontmatter}

\section{Introduction}\label{sec:Introduction}
%\IEEEPARstart{E}NERGY consumption in wireless networks is a vital issue from economic and environmental point of view. To facilitate the way toward green and efficient communication, many sophisticated techniques and technologies are proposed in fifth generation of cellular networks (5G) to enhance EE and yet achieving acceptable spectral efficiency (SE). 
%Multi-user multiple-input multiple-output systems have the potential to increase EE and SE in both uplink and downlink Transmission even if user equipment (UE) is equipped with one antenna. To achieve this so called gain proper beamforming must be performed \cite{Andrews}.

%\IEEEPARstart{T}ILT 
Tilt angle adaptation which is also known as three dimensional beamforming (3DBF), full dimension multiple-input multiple-output (FD-MIMO) or vertical beamforming is a  promising technology for interference management and performance improvement in fifth generation (5G) cellular networks \cite{Andrews,drRazavizadeh}. In this technique by deploying active antenna systems (AAS) at the base station (BS) of cellular networks, it is possible to dynamically adapt the antenna tilt angle in each transmission interval \cite{Lee}. This can be done by modifying parameters of symmetrical 3D pattern introduced in \cite{3GPP}. 
The antenna arrays radiates a fan-shaped beam with a large half power beam width (HPBW) in horizontal plane, while radiating a sharp beam with a small HPBW in the vertical plane which prevents signal from leaking to adjacent cells. 
%The %pattern shows that the 
%antenna array radiates a fan-shaped beam in horizontal plane resulting in a large half power beamwidth (HPBW) while radiating a sharp beam in vertical plane corresponding to a small HPBW which prevents signal to leak to adjacent cells.
%Equation of radiation pattern includes a main lobe and a constant gain region in both planes. 
In contrast to the horizontal plane in which the direction of antenna main lobe is fixed (since the BS orientation is fixed for each sector)%the horizontal plane where the direction of main lobe is fixed (because BS orientation is fixed for each sector)
, the antenna main lobe in the vertical plane can be steered to a desired direction by changing the tilt angle of the BS's antenna array \cite{drRazavizadeh}. 
%by changing the tilt angle of BS's antenna array, main lobe of antenna pattern translates in vertical space % 3D antenna gain introduced in \cite{3GPP} includes a main lobe and a constant gain region. By changing the tilt angle of BS's antenna array, main lobe of antenna pattern translates in vertical space.
%Using information about users' locations and angle of arrival (AoA) provided at the BS, one can aim the main lobe to a specific angle that simultaneously maximizes intra-cell users' gain and minimizes intercell interference on other cell's user equipments (UEs) \cite{Seifi3}. 
To this end, we need information about users' locations and angle of arrival (AoA) at the BS. This tilt angle adjustment can improve some performance metrics such as spectral efficiency (SE), coverage probability or energy efficiency (EE) \cite{Andrews,drRazavizadeh}. 

Most of the previous works on 3DBF study the problem of maximizing of the spectral efficiency in various scenarios. In \cite{Tran} a game theory-based approach is proposed that maximizes the sum SE of a multi-cell network. In addition the authors approximate the signal-to-interference-plus-noise ratio (SINR) by its asymptotic value when the number of the BS antennas goes to infinity. %However, in this paper, the signal-to-interference-plus-noise ratio (SINR) is approximated by it's value in the limit of infinite number of BS antennas.
 In \cite{Seifi1,Seifi2,Seifi3} the authors extract conditional ergodic sum SE of the network in terms of the BS antenna array tilt angle and then approximate it with complicated mathematical expressions. Hence the optimum tilt angles are found using exhaustive search over different possible tilt angles.

In general there are two 3DBF strategies: passive and active. In the passive 3DBF, the antenna pattern cannot be dynamically changed, and thus the antenna array tilt angle is fixed during several transmission intervals. This technique is suitable for applications like self organizing networks \cite{Aquino}. However in the active 3DBF, the antenna pattern can be dynamically changed i.e. the tilt angle is optimized by utilizing instantaneous user's location information and AoA in each transmission interval \cite{drRazavizadeh,Lee}. 
Thus, information about the users' locations is playing a vital role in designing of the 3DBF technique. The users' locations also can be modeled with the help of stochastic geometry (SG) \cite{Lee,Aquino,Shi}. To this end the users' locations is modeled by a uniform distribution and the BSs' are assumed to be located on the vertices of a hexagonal grid \cite{Shi}. In addition, the BSs' locations can also be modeled with Poisson Point Process (PPP) \cite{Aquino}. %Obtaining the optimal tilt angle in a network modeled by SG is called passive 3DBF.
Since the location of the users is a slow varying parameter, it can be estimated almost exactly and the assumption of knowing these informations is a reasonable assumption \cite{Tran, Xiao}. The subject of AoA estimation is well studied in literature \cite{DrHaddadi}. Also there are some valuable works which concern with AoA estimation using novel low energy consuming machine learning approaches such as Bayesian compressive sensing algorithms \cite{Carlin,Wang}. Here similar to \cite{Lee,Tran,Seifi3,Xiao,Heath,Caire} we assume the information of the AoA of the users are accurately available at the BSs.

It is shown that the active 3DBF can achieve higher performance compared with the passive 3DBF since it adapts the antenna pattern to the instantaneous location of the users instead of using an average location of the users \cite{drRazavizadeh,Lee}.
  The problem of maximizing the sum SE in the active and passive 3DBF for a single cell network and two scenarios of  single-user and multi-user is presented in \cite{Lee}. Although in \cite{Lee} potential of the 3DBF in performance improvement and advantage of active 3DBF over \color{black} passive 3DBF are well studied, however because of the single cell arrangement the interference management property of the 3DBF is not fully exploited.
%However, because of the single cell assumption in that paper, the problem of interference management in the 3DBF is not fully investigated.

Although SE is a very important performance metric in the wireless networks, recently EE is also becoming more important from the network service provider's point of view. The 3DBF is one of the techniques that have been proposed to facilitate energy efficient communication in next generation of cellular networks. In this technique, by proper managing of interference and reducing the transmit power, the EE can be increased \cite{drRazavizadeh}. In spite of this fact, the problem of EE maximization by utilizing 3DBF has not been thoroughly investigated in the literature. In \cite{Aquino} by exploiting 3DBF, optimization of the sum EE in a two tier heterogeneous network (HetNet) is addressed.

%Since the location of users is a slow varying parameter in nature, it can be estimated almost exactly. Therefore, the assumption of perfect knowledge of AoA and users' locations at the BS is reasonable. In \cite{Tran}, this information is assumed to be perfectly known at the BS and an algorithm based on game theory is proposed that maximizes sum SE of a multi cell network. However, in this paper, the signal-to-interference-and-noise ratio (SINR) is approximated by it's equivalent in the limit of infinite number of BS antennas. 

%Another way of treating users' location information is to condition network performance metric on BS’s array tilt angle and then approximating ergodic rate with tractable yet complicated mathematical expressions \cite{Seifi1}. Approaches that lie in this category use Monte-Carlo simulations as a way of taking average of rate with respect to users' locations \cite{Seifi2,Seifi3}.

In this paper, we investigate the problem of joint beamforming and tilt angle optimization at the BS antenna array for maximizing the total EE in a multi-cell multi-user network. In fact, on contrary to the previous works on 3DBF which have considered ergodic sum SE averaged over the users' locations, we investigate maximizing the total EE in a more realistic system with information on the users' locations in the BS. %when information of users' location are available at the BS
 In addition, it is worthwhile to note that in all previous works simple linear beamforming such as eigen beamforming or zero-forcing beamforming were considered while we design the optimum beamforming vectors through a non-linear optimization problem. To the best of our knowledge, there is no similar work which jointly design beamfroming vectors and optimize tilt angle.

In addition, in our scenario the BSs cooperate with each other and share their channel vectors to each user. The objective function is total instantaneous EE subject to sum transmit power constraint at the BS. The resulting optimization problem is non-convex and therefore, a fractional programming approach is employed to convert it to a convex problem. To convexify this problem, lemma 1 is introduced and then to solve the convex problem using block coordinate descent method, theorem 1 is introduced. Also to reduce the complexity of searching for finding the optimal tilt angle, the clustering algorithm is proposed in theorem 2.

Although the new problem is convex in terms of the beamforming vectors, it is still non-convex in terms of the tilt angle. To solve this problem, we divide the users into clusters and then in each cluster the optimum beamforming vectors are calculated. Finally, the optimum tilt angle is found according to the cluster in which the beamforming vectors maximize the total instantaneous EE. It will be shown that the optimal beamforming vectors are dependent on both users' elevation angle and channel gains. Simulation results demonstrate that our method outperforms the conventional 2D beamforming in which the antenna array pattern is omnidirectional and only precoding vectors are designed.

Rest of the paper is organized as follows. In Section \ref{sec:SYSTEM MODEL}, the system model and problem formulation are introduced. In Section \ref{sec:3D SOLUTION}, the proposed algorithm for solving the optimization problem is presented. In Section \ref{sec:CLUSTERING}, the clustering algorithm is introduced. Simulation results are presented in Section \ref{sec:SIMULATION}. Finally, Section \ref{sec:CONCLUSION} concludes the paper.  
\setlength{\arrayrulewidth}{1mm}
\setlength{\tabcolsep}{30pt}
\renewcommand{\arraystretch}{1}

{\rowcolors{2}{gray!30!white!50}{gray!5!white!40}
%\begin{centring}
\begin{table}[!h]
%\begin{tabular}{|c|l|}
%\resizebox{\textwidth}{!}{%
\begin{tabular}{!{\vline}p{2.1cm}!{\vline}p{11cm}!{\vline} }

\hline
Parameter & Description\\
\hline
$\theta_j , \phi_j $  & Tilt and boresight angles of the $j$-th BS respectively.   \\
$SLL_{el} , SLL_{az}$ &  Side lobe levels in the elevation and azimuth domain respectively.    \\
$ \boldsymbol x_i$ &  Transmitted signal vector of the $i$-th BS.  \\
$ \boldsymbol \omega_{in}, d_{in}, y_{in} $  &Beamforming vector, data symbol and received signal of the $n$-th user in the $i$-th cell.  \\
$ \boldsymbol g_{ijm}, \beta_{ijm}$ &  Channel vector and large scale factor between the $m$-th user in the $j$-th cell and the $i$-th BS.   \\
$\theta_{3dB} , \phi_{3dB} $ &  Half power beamwidth of the elevation and azimuth patterns, respectively.  \\
$\theta_{ijk} , \phi_{ijk}$ & AoA of the $k$-th user in the $j$-th cell and the $i$-th BS  \\
$R_{jm}$ & Instantaneous rate of the $m$-th user in the $j$-th cell. \\
$R_{max}$ &  Maximum rate achievable if each user receives maximum power of its serving BS and no interference exists.\\
$\boldsymbol W , \boldsymbol \theta$ & Collection of all beamforming vectors and tilt angle of all BSs.\\
$P, P_c, P_0$ & Maximum transmit power available at the BS, RF-chain and constant power consumption of the BS site, respectively.\\
$\xi $ &  Power amplifier inefficiency\\
$M,L,K$ & The number of antennas at the BS, number of cells and the number of the users in each cell.\\
$\eta$ &  Energy efficiency variable.\\
$\epsilon , \delta $ & Thresholds for stop criterion of the inner and outer layer algorithms.\\
$ \tilde d_{jm} , \mu_{jm}, s_{jm}, e_{jm}$ &  Estimated symbol, filter coefficient,
slack variable and symbol estimation error at the UE of the $m$-th user in the $j$-th cell\\
\hline
%\caption{123}
%}
\end{tabular}
%\end{centring}
\caption{List of all parameters throughout the paper.}
\end{table}
}

%\raggedbottom

%\end{itemize}
In the paper scalars are denoted by lower-case letters. Vectors and matrices are denoted by bold-face lower-case and upper-case letters, respectively. $ \prec $ stands for the generalized inequality.
 $\left( . \right) ^{H}$ is the complex conjugate transpose. $\mathbb{E} \left\{ . \right\} $ denotes statistical expectation and $||.||$ is the Euclidean distance (also known as $L^2-\text{norm} $). also $\left( . \right)^*$ represents the complex conjugate norm. $I_M$ is the Identity matrix of size $M$ and $\mathbb{C}^N$ is the $N-\text{dimensional}$ complex vector space. In addition, a list of all variables that are used in our paper is given in table 1.

\color{black}
\section{System Model} \label{sec:SYSTEM MODEL}

We consider a multi-cell network consisting of $L$ cells. Each cell has $K$ users that are uniformly distributed over the cell coverage region and a BS equipped with a linear antenna array consisting of $M$ active elements. Each user is only served by its own cell's BS. 

The user equipment's (UE's) antenna pattern is assumed to be isotropic and the 3D pattern of the BS antennas is modeled as \cite{3GPP} %\cite[page 59]

\begin{equation} \label{3D gain}
\begin{array}{ll}
 \alpha (\theta_{BS,i})= G_{\max}-\min \left[ 12 \left(\frac{\phi}{\phi_{3dB}}\right)^2 ,\text{SLL}_{az}\right]  - \min \left[ 12 \left(\frac{\theta_{BS,i}-\theta_{ijk}}{\theta_{3dB}}\right)^2 ,\text{SLL}_{el}\right].
\end{array}
\end{equation}
%\phi_{BS,i}-\phi_{ijk}
 where as depicted in Fig. \ref{fig:geosetup}, $\theta_{BS,i}$ is the tilt angle of the $i$-th BS which is defined as the angle between the horizon and the main lobe of vertical pattern of the BS antenna array. $\theta_{ijk}$ is the vertical AoA of the $k$-th user in the $j$-th cell to the $i$-th BS.
In addition $\phi =\phi_{BS,i}-\phi_{ijk}$ is the horizontal angle difference between boresight of the $i$-th BS's antenna and the horizontal AoA of the $k$-th user in the $j$-th cell to the $i$-th BS.
It should be noted that $\phi_{BS,i}$ is added to the generic model in \cite{3GPP} to suitably encompass the base stations with nonzero boresight angle.
%\color{black}
%of main lobe pick of vertical pattern with respect to horizon.
% $\phi_{ijk}$ is the horizontal angle of the $k$-th user in the cell $j$ to the $i$-th BS and $\theta_{ijk}$ is the tilt angle of the $k$th user in cell $j$ to antenna array of BS in cell $i$ and horizon.
Furthermore, $\text{SLL}_{az}=25~\text{dB}$ and $\text{SLL}_{el}=20~\text{dB}$ are side lobe levels in the azimuth and elevation planes, respectively. Also $\phi_{3dB} = 65^{\circ}$ and $\theta_{3dB} = 6^{\circ}$ are HPBWs in the azimuth and elevation planes, respectively \cite{Seifi2}. For a conventional 2D beamforming system, the array's gain in the vertical plane is ignored and only horizontal gain and the first constant term is considered.
%\begin{figure}[t!]
  %\includegraphics[width=\linewidth]{basestation3.png}
  %\caption{Multi-cell setup with three rhombus cells and illustration of 3D angles.}
  %\label{fig:geosetup}
%\end{figure}

During the signal transmission in the downlink, each BS broadcasts its signal to all users. This signal is obtained by multiplying the users' symbols in the corresponding beamforming vectors. Hence, the transmit signal from the $i$-th BS is obtained as follows
\begin{equation} \label{transmitsymbs}
\boldsymbol x_i = \sum\limits_{n = 1}^K {\boldsymbol {\omega} _{in}}{d_{in}}
\end{equation}
where $\boldsymbol {\omega} _{in} \in \mathbb C^{M \times 1}$ and $d_{in}$ are the beamforming vector and data symbol of the $n$-th user in the $i$-th cell, respectively. Moreover, the symbols $d_{in}$ are assumed to be uncorrelated with $\mathbb E\{ |d_{in}|^2 \}=1$. %$\mathbb E\{ \mathbf s_i \mathbf {s_i}^{H}\}=\mathbf I_K$ where $\mathbf {s_i} = [s_{i1},s_{i2},...,s_{iK}]$.
The received signal of the $m$-th user in the $j$-th cell can be written as
\begin{equation}
\begin{array}{ll} \label{receivsig}
{y_{jm}} = \sum\limits_{i = 1}^L \sqrt {\alpha \left( {{\theta _{BS, i}}} \right)}{\boldsymbol g_{ijm}^H{\boldsymbol x_i} } = \sum\limits_{i = 1}^L {\sum\limits_{n = 1}^K \sqrt {\alpha \left( {{\theta _{BS, i}}} \right)} {\boldsymbol g_{ijm}^H{\boldsymbol \omega _{in}}{d_{in}}  \; + } }\; {n_{jm}} \hfill \nonumber
\end{array} \end{equation}
\begin{equation} \begin{array}{ll} =\underbrace { \sqrt {\alpha \left( {{\theta _{BS, j}}} \right)} \boldsymbol g_{jjm}^H{\boldsymbol \omega _{jm}}{d_{jm}} }_{\text{desired}} +\underbrace {\sum\limits_{\begin{subarray}{ll}\; i = 1 \\ (i,n) \end{subarray}}^L {\sum\limits_{\begin{subarray}{ll}\; \;n = 1 \\ \ne  (j,m)\end{subarray}}^K \sqrt {\alpha \left( {{\theta _{BS, i}}} \right)} {\boldsymbol g_{ijm}^H{\boldsymbol \omega _{in}}{d_{in}} } }}_{\text{interference}}+\underbrace{n_{jm}}_{\text{noise}},
\end{array}
\end{equation}
where ${\boldsymbol g_{ijm}}\sim \mathcal{CN}\left( {0,{\beta _{ijm}} {\boldsymbol I_M}} \right)$ is the channel vector 
between the $m$-th user in the $j$-th cell and the BS in the $i$-th cell. $\beta _{ijm}$ is the large scale fading factor that includes both path-loss and shadow fading effects. The noise term ${n_{jm}}$ is a circularly symmetric complex gaussian random variable with zero mean and normalized variance i.e. $\mathbb E \{ |{n_{jm}}|^2 \}=1$. 
Instantaneous rate of the $m$-th user in the $j$-th cell can be written as
\begin{equation} \label{rateexpression}
\hat R_{jm}=\log_2\left(1+\frac{\alpha \left( {{\theta _{BS , j}}}\right){\left| {\boldsymbol g_{jjm}^H{\boldsymbol \omega _{jm}}} \right|^2} }{{\sum\limits_{\begin{subarray}\;\; i = 1 \\ (i,n) \end{subarray}}^L {\sum\limits_{\begin{subarray}{ll}\;\;n = 1 \\ \ne (j,m)\end{subarray}}^K {\alpha \left( {{\theta _{BS, i}}} \right)} {\left|\boldsymbol g_{ijm}^H{\boldsymbol \omega _{in}}\right|^2  } }}+1}\right).
\end{equation}
Using logarithm in base $e$, we have $R_{jm}=\ln(1+\text{SINR})=\frac{\hat R_{jm}}{\ln2}$. The overall EE of the network is defined as the total weighted sum SE divided by the total energy consumption in the network as follows
\begin{equation}\label{EEdef}
\text{EE} =\frac{{\sum\limits_{j,m} {{b_{jm}}{R_{jm}}} }}{{\xi {{\sum\limits_{j,m} {\left\| {{\boldsymbol \omega _{jm}}} \right\|} }^2} + ML{P_c} + L{P_0}}} = f\left( {\boldsymbol W,\boldsymbol \theta } \right).
\end{equation}
 In (\ref{EEdef}), $\boldsymbol W$ is the collection of all beamforming vectors and $\boldsymbol \theta = [\theta_{BS,1},\theta_{BS,2},...,\theta_{BS,L}]$ is the vector containing the tilt angles of all BSs. $b_{jm}$ is weight of the $m$-th user in the $j$-th cell and in fact it indicates each user's priority with respect to other users. $P_c$ is power consumption of an active RF-chain and $P_0$ is constant power consumption for maintenance of BS equipments \cite{Arnold}. $\xi \ge 0$ is a constant that accounts for the power amplifiers' inefficiency. Now we can formulate optimization problem of the network as follow
\begin{equation} \label{goalfunc}
\begin{gathered}
\arg \; \mathop {\max }\limits_{\boldsymbol W,\boldsymbol \theta }f\left( {\boldsymbol W,\boldsymbol \theta } \right) =
\frac{{{f_1}\left( {\boldsymbol W,\boldsymbol \theta } \right)}}{{{f_2}\left( {\boldsymbol W, \boldsymbol \theta } \right)}},  \hfill \\
\mathrm{subject~to}\quad  \sum\limits_{m = 1}^K {{{\left\| {{\boldsymbol \omega _{jm}}} \right\|}^2}}  \leqslant P \quad \forall j,  \quad j = 1,2,...,L, \hfill  \\
\qquad \qquad  \qquad{0_{1 \times L}} \prec \boldsymbol \theta  \prec {90^o} \times {\left[ {1,1,...,1} \right]_{1 \times L}}. \hfill 
\end{gathered} 
\end{equation}
%Objective function\footnote{Since the ratio $\frac{R_{jm}}{\hat R_{jm}}$ is a positive number, $\ln2$, so there is no difference in maximizing $R_{jm}=\ln(1+\text{SINR})$ instead of $\hat R_{jm}$} is the overall EE of the network and is the ratio of weighted sum of instantaneous rates to total power consumption in the network.
where $f_1(\boldsymbol W) = {\sum\limits_{j,m} {{b_{jm}}{R_{jm}}} }$ and $f_2(\boldsymbol W) = {\xi {{\sum\limits_{j,m} {\left\| {{\boldsymbol \omega _{jm}}} \right\|} }^2} + ML{P_c} + L{P_0}}$ which are numerator and denominator of equation (\ref{goalfunc}). The first constraint in (\ref{goalfunc}) is the maximum transmit power available at the BS and the second constrain is the dynamic range of the tilt angles. 

\noindent It is possible to show that the following bounds are true for the objective function in (\ref{goalfunc})

\begin{equation}
\begin{cases}
\boldsymbol 0 < {f_1}\left( {\boldsymbol W,\boldsymbol \theta } \right) \leqslant {R_{\max }}  \\
{f_2}\left( {\boldsymbol W,\boldsymbol \theta } \right) \geqslant ML{P_c} + L{P_0}  \\ 
{f_2}\left( {\boldsymbol W,\boldsymbol \theta } \right) \leqslant LP + ML{P_c} + L{P_0} \\
{R_{\max }} = \sum\limits_{j = 1}^L {\sum\limits_{m = 1}^K {{{\log }_2}\left( {1 + \frac{{P{{\left\| {\boldsymbol g_{jjm}^H} \right\|}^2}}}{1}} \right)} }
\end{cases}
\end{equation}
where $R_{max}$ is the maximum rate for the case that the maximum power is allocated for each user and no interference exists.
\begin{figure}[t!]
  \includegraphics[width=\linewidth]{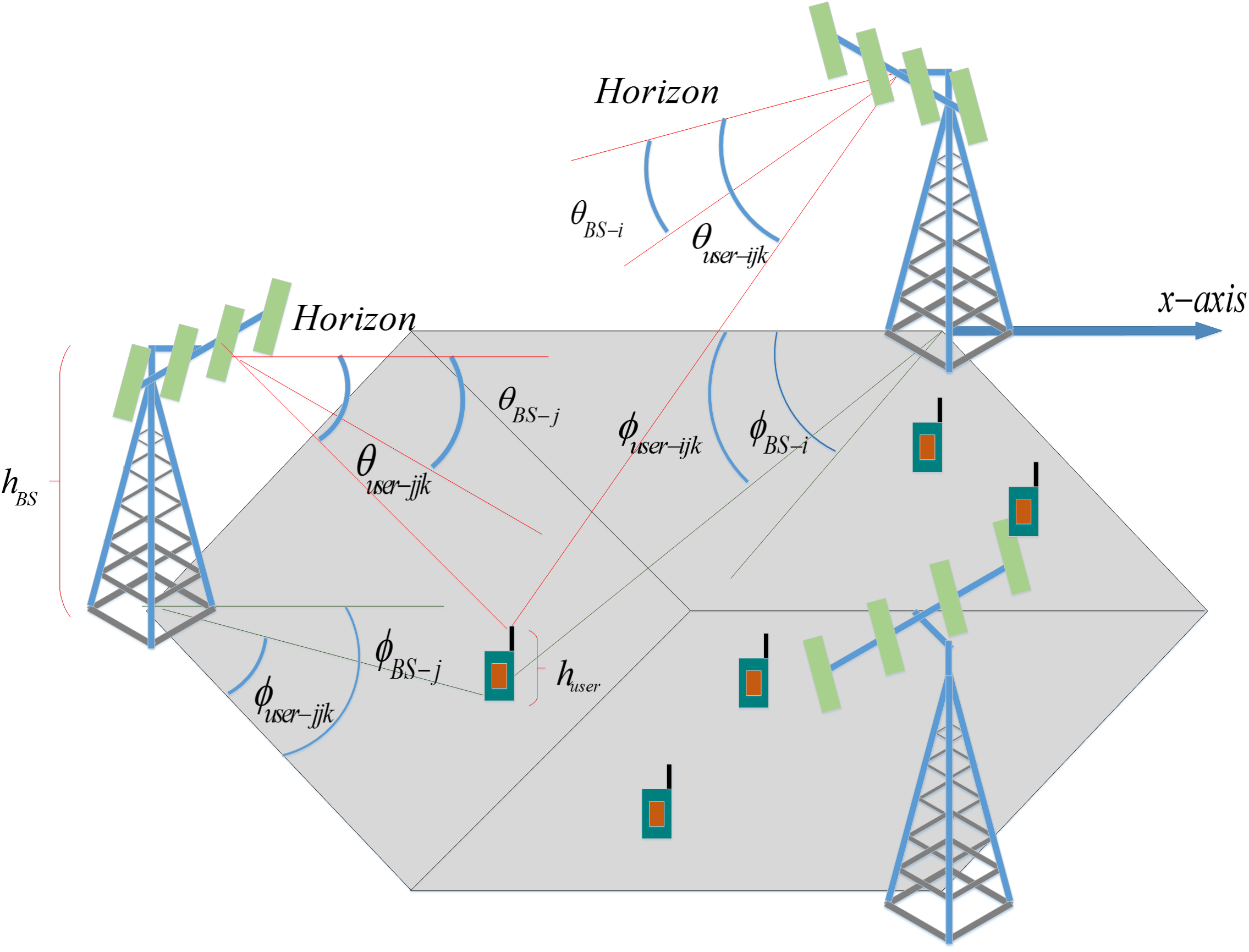}
  \caption{Multi-cell setup with three cells and illustration of 3D angles.}
  \label{fig:geosetup}
\end{figure}
\section{The EE Maximization Problem Solution} \label{sec:3D SOLUTION}

Since the objective function in (\ref{goalfunc}) is fractional and non-convex, fractional programming technique is applied to solve this problem \cite{Crouz}. First consider the following equation
\begin{equation} \label{equivalent1}
F\left( \eta  \right) =  \mathop {\max }\limits_{\left( {\boldsymbol W,\boldsymbol \theta } \right) \in \mathbb D}\left\{ {{f_1}\left( {\boldsymbol W,\boldsymbol \theta } \right) - \eta {f_2}\left( {\boldsymbol W,\boldsymbol \theta } \right)} \right\}
\end{equation}
where $\mathbb D = \{(W,\theta) | \sum\limits_{m = 1}^K {{{\left\| {{\boldsymbol \omega _{jm}}} \right\|}^2}}  \leqslant P \quad \forall j,\; {0_{1 \times L}} \prec \boldsymbol \theta  \prec {90^o} \times {\left[ {1,1,...,1} \right]_{1 \times L}}\}$ is feasible set of the problem. It is shown that $F(\eta)$ in equation (\ref{equivalent1}) has the following properties \cite{Crouz}
\begin{enumerate}[label=(\alph*)]
\item $F(\eta)$ is convex over $\mathbb R$,
\item $F(\eta)$ is continuous over $\mathbb R$,
\item $F(\eta)$ is strictly decreasing,
\item The equation $F(\eta) = 0$ always has a unique solution.
\end{enumerate}
Thus, we can conclude that the following statements are equivalent \cite{Crouz}:
\begin{enumerate}[label=(\alph*)]
\item $\mathop {\max }\limits_{\left({\boldsymbol W,\boldsymbol \theta } \right) \in \mathbb D} f\left( {\boldsymbol W,\boldsymbol \theta } \right) = \mathop {\max }\limits_{\left( {\boldsymbol W,\boldsymbol \theta }\right) \in \mathbb D} \frac{{{f_1}\left( {\boldsymbol W,\boldsymbol \theta } \right)}}{{{f_2}\left( {\boldsymbol W,\boldsymbol \theta } \right)}} = \eta $\\
\item $F\left( \eta  \right) = \mathop {\max }\limits_{\left( {\boldsymbol W,\boldsymbol \theta }\right) \in \mathbb D} \left\{ {{f_1}\left( {\boldsymbol W,\boldsymbol \theta } \right) - \eta {f_2}\left( {\boldsymbol W,\boldsymbol \theta } \right)} \right\} = 0$.
\end{enumerate}
In other words, solving univariate equation $F\left(\eta \right) = 0$ is equivalent to the maximization problem in (\ref{goalfunc}). It means that if one can find an $\eta$ such that the optimum value of problem (\ref{equivalent1}) is zero, then this optimum value will be also a solution to the problem (\ref{goalfunc}). To solve it, the problem can be divided into two parts or layers. The so-called outer layer searches for the optimum value of $\eta$ while the inner layer solves the equation (\ref{goalfunc2}) for each value of $\eta$ to find the optimum beamforming vectors and tilt angle. 

For solving the univariate equation (\ref{equivalent1}), we can employ the bisection algorithm as follows
\makeatletter
\def\BState{\State\hskip-\ALG@thistlm}
\makeatother
\begin{algorithm}
\caption{Outer Layer Algorithm}\label{etabisection}
\begin{spacing}{1.2}
\begin{algorithmic}[1]
\State $\textbf{Initialize:} \quad {\eta _{\min }} = 0, \; {\eta _{\max }} = \frac{{{R_{\max }}}}{{ML{P_c} + L{P_0}}}, \; \varepsilon  = {10^{ - 3}}$
\While {($\left| {{\eta _{\max }} - {\eta _{\min }}} \right| \geqslant \varepsilon$)}
\State  assign $\eta  = \frac{{{\eta _{\min }} + {\eta _{\max }}}}{2}$, now solve (\ref{goalfunc2}) to find the optimum values of $(\boldsymbol W^{opt},\boldsymbol \theta^{opt})$ and calculate $F(\eta)$.
\If {($F(\eta) >0 \; $)} $\;\eta_{max}=\eta$ \Else $\; \eta _{\min}=\eta$
\EndIf
\EndWhile \textbf{end while}
\end{algorithmic}
\end{spacing}
\end{algorithm}

As we see in algorithm 1, for each value of $\eta$, the problem (\ref{goalfunc2}) needs to be solved. This optimization problem can also be written as
\begin{equation} \label{goalfunc2}
\begin{gathered}
\arg \; \mathop {\max }\limits_{\boldsymbol W,\boldsymbol \theta }G\left(\boldsymbol W, \boldsymbol \theta \right) = \left( {\sum\limits_{j,m} {{b_{jm}}{R_{jm}}}  - \eta \xi {{\sum\limits_{j,m} {\left\| {{\boldsymbol \omega _{jm}}} \right\|} }^2}} \right)  \hfill \\
\mathrm{subject~to} \quad  \sum\limits_{m = 1}^K {{{\left\| {{\boldsymbol \omega _{jm}}} \right\|}^2}}  \leqslant P \quad  \forall j,\quad  j = 1,2,...,L \hfill \\
\qquad \qquad \quad \quad{0_{1 \times L}} \prec \boldsymbol \theta  \prec {90^o} \times {\left[ {1,1,...,1} \right]_{1 \times L}} \hfill
\end{gathered}
\end{equation}
In order to solve (\ref{goalfunc2}), we use the method introduced in \cite{Cioffi}. According to this method, the data rate of users $R_{jm}$ can be written in terms of achievable minimum mean square error (MMSE) of estimated data symbols in the UE's. This can be done by introducing optimum filter variables $ \mu _{jm}$ which are defined by the following equation
\begin{equation}\label{filtervar}
 {\tilde d_{jm}} = \mu _{jm}^*{y_{jm}} 
\end{equation}

 where the $\tilde d_{jm}$ is the estimated data symbol in the $m$-th user in the $j$-th cell. For deriving the MSE, let us introduce $\tilde e_{jm}$ as

 \begin{equation}
\begin{gathered}
\tilde e_{jm} = \left( {\mu _{jm}^*\sum\limits_{i = 1}^L {\sum\limits_{n = 1}^K {\boldsymbol g_{ijm}^H{\boldsymbol \omega _{in}}{d_{in}}\sqrt {\alpha \left( {{\theta _{BS , i}}} \right)}  + } } {n_{jm}} - {d_{jm}}} \right). \hfill 
\end{gathered}
\end{equation}
  Then we can write MSE in terms of the receiving filter coefficients and channel parameters as

\begin{equation} \label{MMSE}
\begin{gathered}
{e_{jm}} = \mathbb{E}\left\{ {\left( {{{\tilde d}_{jm}} - {d_{jm}}} \right){{\left( {{{\tilde d}_{jm}} - {d_{jm}}} \right)}^*}} \right\} = \mathbb{E}\left\{ {\tilde e_{jm}{\tilde e_{jm}^*}} \right\} \hfill
\\
 = {\left| {{\mu _{jm}}} \right|^2}\left( {\sum\limits_{i = 1}^L {\sum\limits_{n = 1}^K {{{\left| {\boldsymbol g_{ijm}^H{\boldsymbol \omega _{in}}} \right|}^2} {\alpha \left( {{\theta _{BS , i}}} \right)}  + } } 1} \right) \hfill \\
- \mu _{jm}^*\boldsymbol g_{jjm}^H{\boldsymbol \omega _{jm}}\sqrt {\alpha \left( {{\theta _{BS , j}}} \right)}  - {\mu _{jm}}\boldsymbol \omega _{jm}^H{\boldsymbol g_{jjm}}\sqrt {\alpha \left( {{\theta _{BS , j}}} \right)}  + 1. \hfill
\end{gathered}
\end{equation}

In the sequel, we will use the following lemma \cite{Bore}:
\theoremstyle{definition}\newtheorem{lemma}{Lemma}
\begin{lemma} \label{lemma}
Maximum value of function $f\left(\boldsymbol S \right) =  - \text{Tr}\left( {\boldsymbol{SE}} \right) + \log \left( {\det \left(\boldsymbol S \right)} \right) + d$ which is the solution of following optimization problem. The solution of
%\begin{equation*}
$$\mathop {\max }\limits_{\boldsymbol S \in {\mathbb{C}^{d \times d}},\boldsymbol S \succ 0}f\left(\boldsymbol  S \right)$$
%\end{equation*}

\noindent exists at
$${\boldsymbol S^{opt}} = {\boldsymbol E^{ - 1}}$$ 
and the optimum value is
$$f\left( {{\boldsymbol S^{opt}}} \right) = \log \left( {\det \left( {{\boldsymbol E^{ - 1}}} \right)} \right).$$
\end{lemma}

%\begin{proof}
%This is an unconstrained optimization problem \textbf{and easily by} taking derivative with respect to $\boldsymbol S$ and finding zeroes of derivative, optimum solution and maximum value will be found
%$$\frac{\partial f\left(\boldsymbol S \right)}{\partial \boldsymbol S}=-\boldsymbol E+{{\boldsymbol S}^{-1}}=0 \quad \Rightarrow \boldsymbol S^{opt}=\boldsymbol E^{-1}$$
%\end{proof}

By using lemma \ref{lemma}, the optimization problem in (\ref{goalfunc2}) can be rewritten as 
\begin{multline} \label{goalfunc3}
\mathop {\max }\limits_{\boldsymbol W,\boldsymbol U,\boldsymbol S,\boldsymbol \theta }  {\rm H}\left( {\boldsymbol W,\boldsymbol U,\boldsymbol S,\boldsymbol \theta } \right) \hfill \\
\text{subject to}  \; \sum\limits_{m = 1}^K {{{\left\| {{\boldsymbol \omega _{jm}}} \right\|}^2}}  \leqslant P  \quad \forall j,  j = 1,2,...,L  \hfill \\
\qquad \qquad \quad{\boldsymbol 0_{1 \times L}} \prec \boldsymbol \theta  \prec {90^o} \times {\left[ {1,1,...,1} \right]_{1 \times L}} \hfil 
\end{multline}
where new objective function $\rm {\rm H}\left( {\boldsymbol W,\boldsymbol U,\boldsymbol S,\boldsymbol \theta } \right)$ and set of filter variables $\boldsymbol U$ slack variables $\boldsymbol S$ are defined as
\begin{multline} \label{defgoal3}
{\rm H}\left( {\boldsymbol W,\boldsymbol U,\boldsymbol S,\boldsymbol \theta } \right)= \sum\limits_{j,m} {\left( { - {b_{jm}}{e_{jm}}{s_{jm}} + {b_{jm}}\log \left( {{s_{jm}}} \right)} \right)}  
\hfill \\+ \sum\limits_{j,m} {\left( {{b_{jm}} - \eta \xi {{\left\| {{\boldsymbol \omega _{jm}}} \right\|}^2}} \right)} \hfill \\
\boldsymbol U = \left\{ {{\mu _1},{\mu _2},...,{\mu _L}} \right\} \quad \text{and} \quad{\mu _j} = \left\{ {{\mu _{j1}},{\mu _{j2}},...,{\mu _{jK}}} \right\} \hfill \\
\boldsymbol S = \left\{ {{s_1},{s_2},...,{s_L}} \right\} \quad \; \; \text{and} \quad {s_j} = \left\{ {{s_{j1}},{s_{j2}},...,{s_{jK}}} \right\} \hfill
\end{multline}
Now by substituting (\ref{MMSE}) in (\ref{goalfunc3}), we get 
\begin{equation} \label{goalfunc4}
\begin{gathered}
\mathop {\max }\limits_{{\boldsymbol W_j},{\theta _j}} ( - \sum\limits_{m = 1}^K {\sum\limits_{i = 1}^L {\sum\limits_{n = 1}^K {{b_{in}}{s_{in}}{{\left| {{\mu _{in}}} \right|}^2}{{\left| {\boldsymbol g_{jin}^H{\boldsymbol \omega _{jm}}} \right|}^2}\alpha \left( {{\theta _{BS , j}}} \right)} } }  \hfill \\
+ \sum\limits_{m = 1}^K {{b_{jm}}{s_{jm}}\mu _{jm}^*\boldsymbol g_{jjm}^H{\boldsymbol \omega _{jm}}\sqrt {\alpha \left( {{\theta _{BS , j}}} \right)} } \; ) \hfill \\+ ( \sum\limits_{m = 1}^K {{b_{jm}}{s_{jm}}{\mu _{jm}}\boldsymbol \omega _{jm}^H{\boldsymbol g_{jjm}}\sqrt {\alpha \left( {{\theta _{BS , j}}} \right)} }  - \eta \xi {\left\| {{\boldsymbol \omega _{jm}}} \right\|^2}\; ) \hfill \\
\text{subject to}  \; \sum\limits_{m = 1}^K {{{\left\| {{\boldsymbol \omega _{jm}}} \right\|}^2}}  \leqslant P  \quad \forall j,  j = 1,2,...,L \qquad \hfill \\ \qquad \qquad \; \;{\boldsymbol 0_{1 \times L}} \prec \boldsymbol \theta  \prec {90^o} \times {\left[ {1,1,...,1} \right]_{1 \times L}} \hfill
\end{gathered}
\end{equation}
By substituting (\ref{MMSE}) in (\ref{goalfunc3}) and (\ref{defgoal3}), the optimization problem of (\ref{goalfunc3}) decouples among the BSs and each BS must solve its own problem. However,  sharing of filter coefficients and slack variables through backhaul links is still necessary.
\theoremstyle{definition}\newtheorem{theo}{Theorem}
\begin{theo} \label{theorem1}
For each value of $\boldsymbol W$ and $\boldsymbol \theta$, the optimum values of $\boldsymbol U$ and $\boldsymbol S$ which maximizes (\ref{goalfunc4}) are
\begin{multline} \label{optfilt}
\mu _{jm}^{opt} = \frac{{\boldsymbol g_{jjm}^H{\boldsymbol \omega _{jm}}\sqrt {\alpha \left( {{\theta _{BS , j}}} \right)} }}{{\sum\limits_{i = 1}^L {\sum\limits_{n = 1}^K {{{\left| {\boldsymbol g_{ijm}^H{\boldsymbol \omega _{in}}} \right|}^2}\sqrt {\alpha \left( {{\theta _{BS , i}}} \right)}  + } } 1}}\hfill \\
\frac{1}{{{{\hat e}_{jm}}}} = 1 - \frac{{{{\left| {\boldsymbol g_{jjm}^H{\boldsymbol \omega _{jm}}} \right|}^2}\alpha \left( {{\theta _{BS , j}}} \right)}}{{\sum\limits_{i = 1}^L {\sum\limits_{n = 1}^K {{{\left| {\boldsymbol g_{ijm}^H{\boldsymbol \omega _{in}}} \right|}^2}\sqrt {\alpha \left( {{\theta _{BS , i}}} \right)}  + } } 1}}  \\ s_{jm}^{opt} = \frac{1}{{{{\hat e}_{jm}}}} \hfill
\end{multline}
\end{theo}
\begin{proof}
Similar to lemma \ref{lemma}, to find the optimum value, two equations $\frac{\partial {\rm H}\left( {\boldsymbol W,\boldsymbol U,\boldsymbol S,\boldsymbol \theta } \right)}{\partial s_{j,k}}=0$  and $\frac{\partial {\rm H}\left( {\boldsymbol W,\boldsymbol U,\boldsymbol S,\boldsymbol \theta } \right)}{\partial \mu^*_{j,k}}=0$ must be solved which results in (\ref{optfilt}).
\end{proof}
Clearly as Theorem \ref{theorem1} states, if we have optimum values of $\boldsymbol W$ and $\boldsymbol \theta$ then the optimum values of $\boldsymbol U$ and $\boldsymbol S$ are derived. The optimization problem of (\ref{goalfunc3}) is equivalant to finding the optimum values of $\boldsymbol W$ and $\boldsymbol \theta$ in (\ref{goalfunc4}). 
Now with help of Theorem \ref{theorem1}, we aim to solve (\ref{goalfunc4}) using its dual problem (Lagrangian method) and block coordinate descent method.

Now we construct Lagrangian dual problem as
\begin{equation}\label{lagrange}
\begin{gathered}
L\left( {{\boldsymbol W_j},{\lambda _j},{\theta _{BS,j}}} \right) = \hfill \\ - \sum\limits_{m = 1}^K {\left( {\sum\limits_{i = 1}^L {\sum\limits_{n = 1}^K {{b_{in}}{s_{in}}{{\left| {{\mu _{in}}} \right|}^2}{{\left| {\boldsymbol g_{jin}^H{\boldsymbol \omega _{jm}}} \right|}^2}\alpha \left( {{\theta _{BS , j}}} \right) - {b_{jm}}{s_{jm}}\mu _{jm}^*\boldsymbol g_{jjm}^H{\boldsymbol \omega _{jm}}\sqrt {\alpha \left( {{\theta _{BS , j}}} \right)} } } } \right)} \hfill \\
+ \left( {\sum\limits_{m = 1}^K {\left( {{b_{jm}}{s_{jm}}{\mu _{jm}}\boldsymbol \omega _{jm}^H{\boldsymbol g_{jjm}}\sqrt {\alpha \left( {{\theta _{BS , j}}} \right)}  - \eta \zeta {{\left\| {{\boldsymbol \omega _{jm}}} \right\|}^2}} \right)} } \right) - {\lambda _j}\left( {\sum\limits_{m = 1}^K {{{\left\| {{\boldsymbol \omega _{jm}}} \right\|}^2}}  - P} \right). \hfill
\end{gathered}
\end{equation}
Using these two equalities \cite{Boyd}
\begin{multline} \nonumber
\frac{{\partial {{\left\| {{\boldsymbol \omega _{jm}}} \right\|}^2}}}{{\partial {\boldsymbol \omega _{jm}}}} = 2{\boldsymbol \omega _{jm}}, \hfill \\
\frac{{\partial {{\left| {\boldsymbol g_{jin}^H{\boldsymbol \omega _{jm}}} \right|}^2}}}{{\partial {\boldsymbol \omega _{jm}}}} = \frac{{\partial \left( {\boldsymbol \omega _{jm}^H{\boldsymbol g_{jin}}\boldsymbol g_{jin}^H{\boldsymbol \omega _{jm}}} \right)}}{{\partial {\boldsymbol \omega _{jm}}}} = 2{\boldsymbol g_{jin}}\boldsymbol g_{jin}^H{\boldsymbol \omega _{jm}} \hfill
\end{multline}
%\color{black}
and also using KKT conditions, the optimum solution of (\ref{lagrange}) for the beamforming vectors at a given $\theta$ can be calculated as follows
\begin{multline}
\frac{{\partial L\left( {{\boldsymbol W_j},{\lambda _j},{\theta _j}} \right)}}{{\partial {\boldsymbol \omega _{jm}}}} =  - 2\sum\limits_{i = 1}^L {\sum\limits_{n = 1}^K {{b_{in}}{s_{in}}{{\left| {{\mu _{in}}} \right|}^2}{\boldsymbol g_{jin}}\boldsymbol g_{jin}^H{\boldsymbol \omega _{jm}}\alpha \left( {{\theta _{BS - j}}} \right)} } \hfill \\
 + 2{b_{jm}}{s_{jm}}{\mu _{jm}}{\boldsymbol g_{jjm}}\sqrt {\alpha \left( {{\theta _{BS - j}}} \right)}  - 2\eta \zeta {\boldsymbol \omega _{jm}} = 0 \hfill \\
\Rightarrow \left( {\sum\limits_{i = 1}^L {\sum\limits_{n = 1}^K {{b_{in}}{s_{in}}{{\left| {{\mu _{in}}} \right|}^2}{\boldsymbol g_{jin}}\boldsymbol g_{jin}^H\alpha \left( {{\theta _{BS - j}}} \right)} }  + \eta \zeta {\boldsymbol I_M}} \right){\boldsymbol \omega _{jm}} = {b_{jm}}{s_{jm}}{\mu _{jm}}{\boldsymbol g_{jjm}}\sqrt {\alpha \left( {{\theta _{BS - j}}} \right)}.
\end{multline}
This yields a solution as follows
\color{black}

\begin{multline} \label{beamformer}
\frac{{\partial L\left( {{\boldsymbol W_j},{\lambda _j},{\theta _j}} \right)}}{{\partial {\boldsymbol \omega _{jm}}}} = 0   \\ \Rightarrow   {\boldsymbol \omega _{jm}} = {b_{jm}}{s_{jm}}{\mu _{jm}}{\left( {{\boldsymbol A_j} + {\lambda _j}{\boldsymbol I_M}} \right)^\dag }{\boldsymbol g_{jjm}}\sqrt {\alpha \left( {{\theta _{BS , j}}} \right)} \hfill \\ \Rightarrow   {\boldsymbol \omega _{jm}} = {b_{jm}}{s_{jm}}{\mu _{jm}}{\left( {{\boldsymbol A_j} + {\lambda _j^{*}}{\boldsymbol I_M}} \right)^\dag }{\boldsymbol g_{jjm}}\sqrt {\alpha \left( {{\theta _{BS , j}}} \right)} \hfill
\end{multline}
where ${\boldsymbol A_j} = \sum\limits_{i = 1}^L {\sum\limits_{n = 1}^K {{b_{in}}{s_{in}}{{\left| {{\mu _{in}}} \right|}^2}{\boldsymbol g_{jin}}\boldsymbol g_{jin}^H\alpha \left( {{\theta _{BS , j}}} \right) + \eta \zeta {\boldsymbol I_M}} }$.
In (\ref{beamformer}) the optimum value of Lagrangian's multiplier $\lambda_j$ must be found through one dimensional search which is a common method in optimization via dual problem \cite{He,Cioffi}. We denote this optimum value with $\lambda_j^{*}$. As stated earlier, we see in (\ref{beamformer}) that optimum values of beamforming vectors are dependent on both channel gains and users' elevations. Therefore finding optimum value of the tilt angle is inevitable. Unfortunately (\ref{lagrange}) is a non-convex function in terms of the tilt angle. One way to find the optimum value of the tilt angle is through one dimensional search but this induces a heavy computational task on network. In the next section a new algorithm is proposed for reducing the complexity significantly. 
\section{Clustering Algorithm} \label{sec:CLUSTERING}
There are some local maximums for $\theta$ in the maximization problem (\ref{lagrange}). Finding the global maximum is possible by searching feasible set of tilt angle $\left[ {{\theta }_{\max}},{{\theta }_{\min}} \right]$ where ${{\theta }_{\max}}$ and ${{\theta }_{\min}}$ are the maximum and minimum AoA of the users in the cell, respectively \cite{Lee,Seifi1,Seifi2,Seifi3}. This approach is reasonable when we optimize an ergodic utility which varies very slowly with time. However this is not the case in our setup since channel of the users is varying with time and we need to maximize (\ref{lagrange}) in each transmission interval. In the following theorem, we consider reducing this search interval.

\setcounter{MYtempeqncnt}{1}

\theoremstyle{definition}\newtheorem{theorem2}[MYtempeqncnt]{Theorem}
\begin{theorem2}\label{cluster}
Optimum value of $\theta_{BS,j}$ in (\ref{lagrange}) exists in a symmetric interval of length $2\times \frac{{{\theta }_{3dB}}}{\sqrt{2.4\ln 10}}$ centered around one of users' AoA.
\end{theorem2}
\begin{proof}
As we know, at the location of maximum of a function, its second derivative is negative. Furthermore (\ref{lagrange}) is weighted sum of shifted versions of (\ref{3D gain}). By taking the second derivative of (\ref{3D gain}) with respect to $\theta_{BS,j}$, we can find the interval in which the second derivative of the antenna pattern is negative and hence this interval contains the maximum of (\ref{lagrange}). This is due to the fact that only in a specific interval around AoA of each user the second derivative of (\ref{lagrange}) is negative. In othe words we have
\begin{equation}
\begin{gathered}
\frac{d^2 a(\theta_{BS})}{d \theta_{BS}^2}={{A}_{max}}\left( \frac{{{\left( 2.4\ln 10 \right)}^{2}}}{\theta _{3dB}^{4}}{{\left( {{\theta }_{BS}}-{{\theta }_{user}} \right)}^{2}}-\frac{2.4\ln 10}{\theta _{3dB}^{2}} \right) \hfill \\ \times \exp \left( -1.2\times \frac{{{\left( {{\theta }_{BS}}-{{\theta }_{user}} \right)}^{2}}}{\theta _{3dB}^{2}}{{\log }_{2}}10 \right)\le 0 \hfill \\
\Rightarrow \frac{{{\left( 2.4\ln 10 \right)}^{2}}}{\theta _{3dB}^{4}}{{\left( {{\theta }_{BS}}-{{\theta }_{user}} \right)}^{2}}-\frac{2.4\ln 10}{\theta _{3dB}^{2}}\le 0 \hfill \\
\Rightarrow {{\theta }_{user}}-\frac{{{\theta }_{3dB}}}{\sqrt{2.4\ln 10}}\le {{\theta }_{BS}}\le {{\theta }_{user}}+\frac{{{\theta }_{3dB}}}{\sqrt{2.4\ln 10}} \hfill
\end{gathered}
\end{equation}
\end{proof}
Fig. \ref{fig:clustering} represents this fact for $\theta_{3dB}=6^{\circ}$ and two user case with equal gains in which the array's maximum gain is normalized to unity. Instead of searching whole feasible set for $\theta$, only intervals spanned by clusters will be investigated. Clustering algorithm is presented as follow: \\

\makeatletter
\def\BState{\State\hskip-\ALG@thistlm}
\makeatother
\begin{algorithm}
\caption{Clustering Algorithm}\label{cluster}
\begin{spacing}{1.2}
\begin{algorithmic}[1]
\State Sort angles of users in each cell.
\State Put all users with angle difference less than $2 \times \frac{ {{\theta }_{3dB}}}{\sqrt{2.4\ln 10}}$ in the same cluster.
\State In each cluster users with maximum and minimum AoA represent cluster interval
\If {(a cluster contains just a single user)} only AoA of that user represents that cluster \nonumber
\EndIf
\end{algorithmic}
\end{spacing}
\end{algorithm}
Outputs of algorithm (\ref{cluster}) are tilt angle spanning interval of each cluster. It is enough to investigate spanning interval of clusters instead of whole feasible region of the tilt angle to find global maximum of (\ref{lagrange}) for tilt angle. To further reduce computational task, first we consider the tilt angle be equal to AoA of users and we choose the AoA of the user that maximizes (\ref{lagrange}). We call the user that its AoA maximizes (\ref{lagrange}) compared to other users as chosen user and the cluster containing him as chosen cluster. Then global maximum for tilt angle exists in the chosen. In order to find this global maximum, chosen cluster must be investigated. with this simplification it is enough to search spanning interval of the chosen cluster instead of whole feasible set of tilt angle. The overall algorithm for solving (\ref{goalfunc2}) is shown in Algorithm 3 and its flowchart is presented in Fig.~2.%\label{fig:flowchart}.
\addtocounter{algorithm}{-1}
\makeatletter
\def\BState{\State\hskip-\ALG@thistlm}
\makeatother
\begin{algorithm}
\caption{Inner Layer Algorithm}\label{inner}
\begin{spacing}{1.4}
\begin{algorithmic}[1]
\State \textbf{Initialize:} initialize step counter $n=0$, beamforming vectors such that$ \sum\limits_{m=1}^{K}{{{\left\| \omega _{jm}^{n} \right\|}^{2}}}\le P$, filter and slack variables $\mu _{jm}^{n}=0, \quad s_{jm}^{n}=0$, quit threshold $\delta = 10^{-3}$
\State run clustering algorithm and acquire clusters
%\State $G_{old}=0$ and $G_{new}=G\left(\boldsymbol W^{n}, \boldsymbol \theta^{n} \right)$
\While{$(\left| G\left(\boldsymbol W^n, \boldsymbol \theta^n \right) -G\left(\boldsymbol W^{n-1}, \boldsymbol \theta^{n-1} \right)  \right|\ge \delta )$}
%\While{$(\left| G_{new} -G_{old}  \right|\ge \delta )$}
\State $(n+1) \rightarrow n$ and update filters $\mu_{jm}^n$ using (\ref{optfilt})
\State update slack variables $s_{jm}^n$ using (\ref{optfilt})
\State find chosen user (the user that considering its AoA as tilt, maximizes $G\left(\boldsymbol W, \boldsymbol \theta \right)$ compared to other user) and 
\For{($i=0$ to $i= \text{length}(\text{chosen\_user\_cluster})$)}
\State $\theta_j= \text{chosen\_user\_cluster}(i)$
\State update bemforming vectors $\boldsymbol \omega_{jm}^n$ using (\ref{beamformer})
\State update $G_{\text{new}}$ using (\ref{goalfunc2})
\If{(${{G}_{new}}>G\left(\boldsymbol W^n, \boldsymbol \theta^n \right) $)} $\text{tilt}_j^{opt}=\quad \text{possible\_tilts}_j(i)$ and $\boldsymbol W_j^{opt}=\boldsymbol W^{n}$
\EndIf
\EndFor \textbf{end for}
\EndWhile \textbf{end while}
\caption{\newline chosen\_ user\_ cluster(i) is the $i$-th angle in the cluster containing chosen user (we examine each cluster spanning with known steps, for example $0.1^{\circ}$ steps)}
\end{algorithmic}
\end{spacing}
\end{algorithm}

\begin{figure}[t!]
  \includegraphics[width=\linewidth]{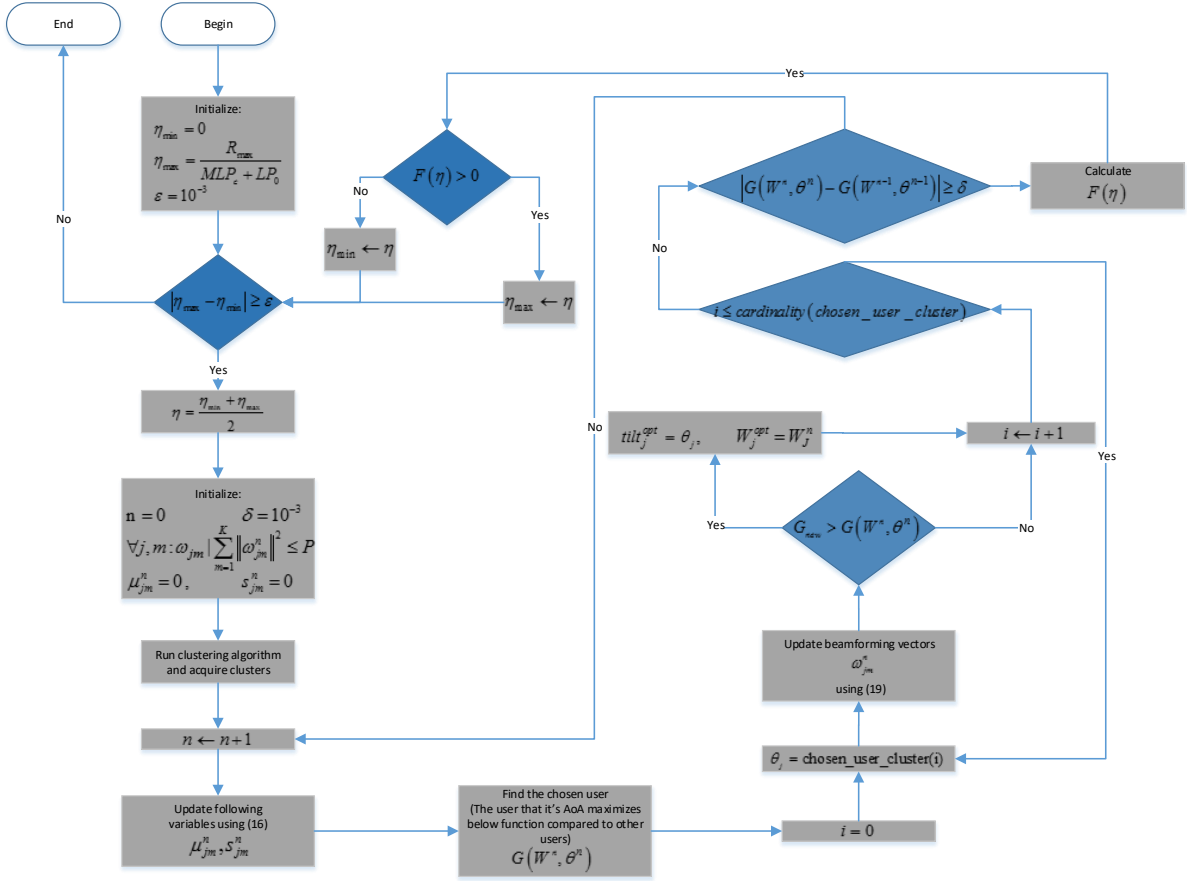}
  \caption{Flowchart of the proposed algorithms.}
  \label{fig:flowchart}
\end{figure}

\color{black}

\begin{figure}[b!]
  \includegraphics[width=\linewidth]{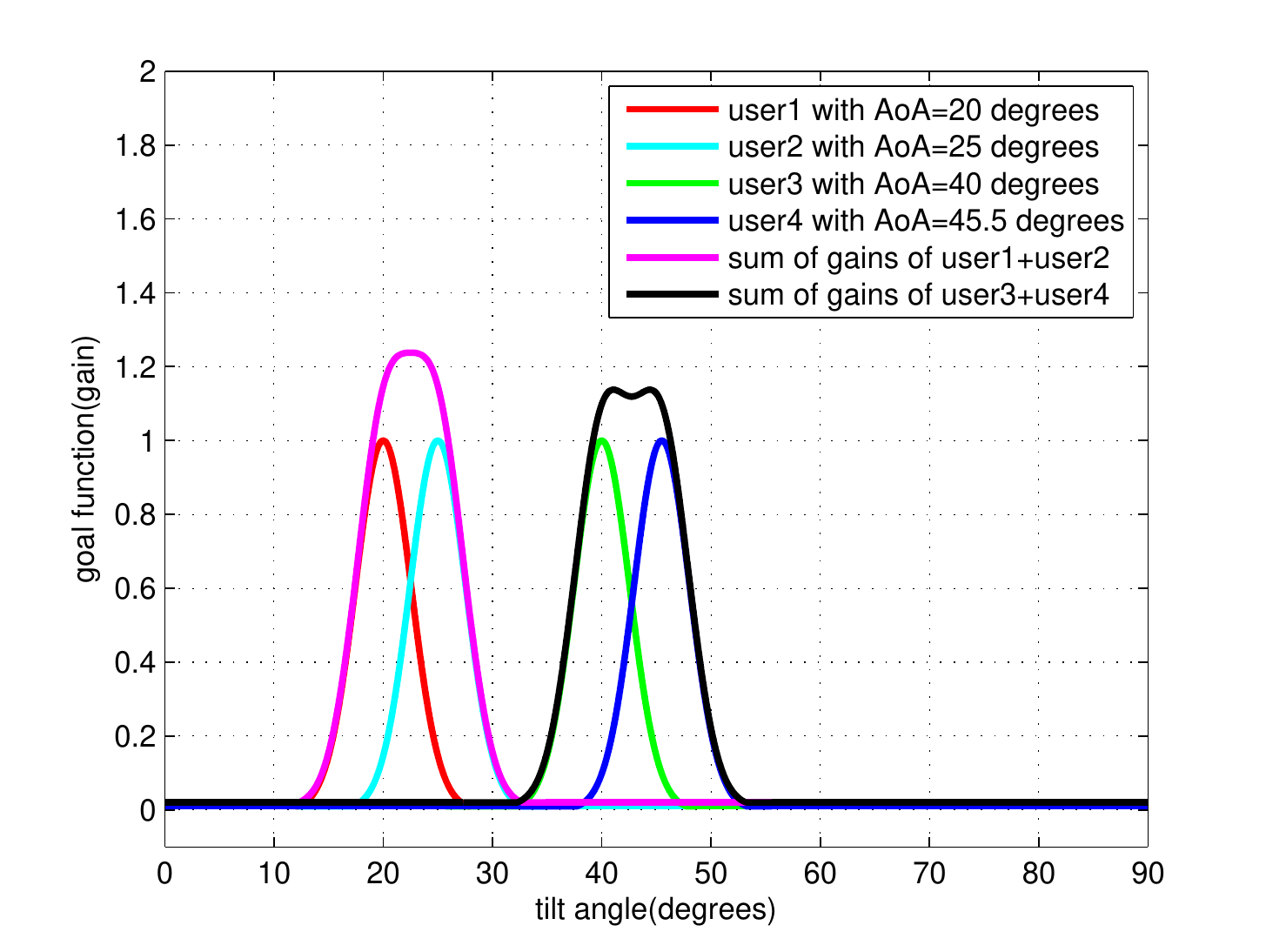}
  \caption{Clustering users based on their angle of arrival difference. (For example with $\theta_{3dB}=6^{\circ}$ clustering interval length will be $\frac{{{\theta }_{3dB}}}{\sqrt{2.4\ln 10}}\approx5.3\ $)}
  \label{fig:clustering}
\end{figure}

\section{Simulation Results} \label{sec:SIMULATION}
 We consider three cooperative adjacent cells as main interferers, i.e. $L=3$. Parameters of simulations are extracted from \cite{3GPP,Arnold}. Height of BSs and UEs are $h_{BS}=32 m$ and $h_{UE}=1.5 m$, respectively. Cell radius is set to be $R=500m$. Channel vector model is $\boldsymbol g_{ijm}= \sqrt{\beta _{ijm}}\boldsymbol h_{ijm}$ where $\boldsymbol h_{ijm}$ is i.i.d. white gaussian random vector that models small scale fading and large scale fading coefficient ${\beta _{ijm}} = \frac{{{z_{ijm}}}}{{d_{ijm}^v}}$ compromise pathloss part ${d_{ijm}^v}$ with $v=3.8$ and shadow fading effect with log-normal random variable $z_{ijm}$ with zero mean and $8$ dB standard deviation. The RF-chain power consumption is $P_c=30$ dB and BS site maintenance power consumption is $P_0=40$. Power amplifier inefficiency is considered to be $\xi = 1$ and all users have identical priority i.e. $b_{jm}=1 \; \forall j,m$. Threshold values for stopping criteria of algorithms are set to $\delta = 10^{-3}$ and $\varepsilon = 10^{-3}$. Also for taking average over user configuration and random channel realization, 2500 Monte-Carlo simulations are done and the results are average of these simulations. Transmit power constraint of the BSs' is considered to vary in the interval $22 \sim 50$ dBm.
%\iffalse
\begin{figure}[!h]
  \includegraphics[width=\linewidth]{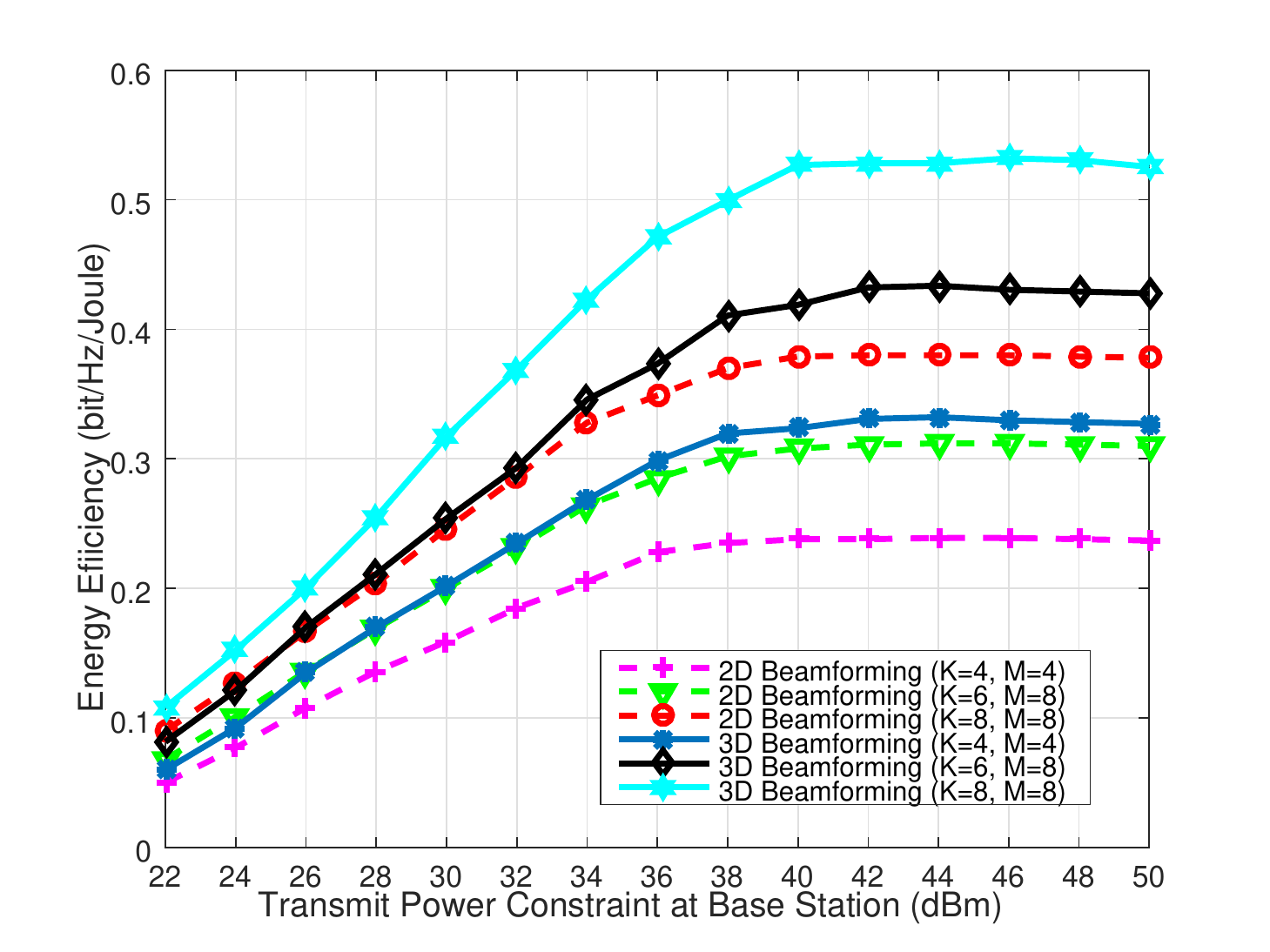}
  \caption{Comparison of 3D beamforming and conventional 2D beamforming for $M=4$ and $M=8$ number of antennas.}
  \label{fig:wholecases}
\end{figure}
%\fi

In Fig.~\ref{fig:wholecases} the total instantaneous EE for 3D case and its corresponding 2D case is presented. As it is expected by increasing the maximum transmit power in the BSs, the EE increases until it reaches a saturated level and after that by increasing available power the EE does not show variations. This is due to this fact that in low SNR regime (noise limited regime) logarithm function grows faster than linear function, thus by increasing transmit power in BS the EE will increase. On the other hand in high SNR regime (interference limited regime), logarithm function growth is lower than that of linear function, hence algorithm will not use the excess transmit power available at the BS and EE converges to a saturated level. It is observed, the 3D beamforming outperforms conventional 2D beamforming with isotropic BS's antenna array gain in the vertical plane. 

Moreover by increasing the transmit power of the BSs, 3D beamforming performance gain over 2D beamforming increases as shown in Fig. \ref{fig:wholecasesgain}. This shows the potential of the 3D beamforming method for mitigating intercell interference. After increasing the available transmit power to the level needed to overcome noise (crossing noise limited regime), limiting factor will become interference. The 3D beamforming gain over 2D beamforming increases because of the ability of interference mitigation. In other words 3D beamforming is more robust to intercell interference compared with 2D beamforming.

Finally Fig. \ref{fig:exhaustive} shows performance comparison of the proposed clustering algorithm and exhaustive search over tilt angle feasible region. 
The exhaustive search algorithm solves the beamforming problem by exploring over all possible values of tilt angles in all the BSs which results an exponential complexity. By contrast, the clustering algorithm reduces the complexity of the searches by limiting the interval of possible tilt angles in the BSs. In fact, in the clustering algorithm only some candidate angles are investigated.%The exhaustive search solves the beamforming problems (inner and outer algorithms) for the complete range of tilt angles of all the BSs while clustering algorithm solves the beamforming problem only for tilt angles that can be optimal (i.e. the tilt angles which are obtained from the proposed clustering algorithm in page 13). This means clustering algorithm greately reduces the complexity of joint tilt angle adaptation and beamforming method. 

We see that their performance is almost equal and a little gap appears between performances in high transmit power region. This is due to approximation considered in clustering algorithm when the chosen cluster contains only one user. The optimum tilt angle in clustering algorithm in this case is considered to be tilt angle of the chosen user, while in fact the optimum tilt angle can be around this user and an interval around this user should be searched. But as we see this approximation causes a little gap compared to the optimum value but the computational complexity gain is considerable. It is notable that the performance gap appears in the interference limited regime where objective function is more sensitive to deviations from the optimum value while in noise limited regime sensitivity to deviations is low and hence there is no gap and no performance loss between two strategies.

Fig.~\ref{fig:userincrement} shows the effect of increasing the number of the users in each cell. As it can be seen, by increasing the number of the users total EE will increase in system. This implies that by adding more users, interference management capability of proposed method over the 2D technique will increase. However it is worth mentioning that by increasing the number of cells, performance gain will shrink. This is due to the fact that the number of degrees of freedom is controlled by the number of antennas and increasing number of cells only increases interference. Hence the system performance degrades. But this performance loss is insignificant due to sectoring cells and side lobe level loss of antenna patterns which heavily suppresses interference power.
\begin{figure}[!h]
  \includegraphics[width=\linewidth]{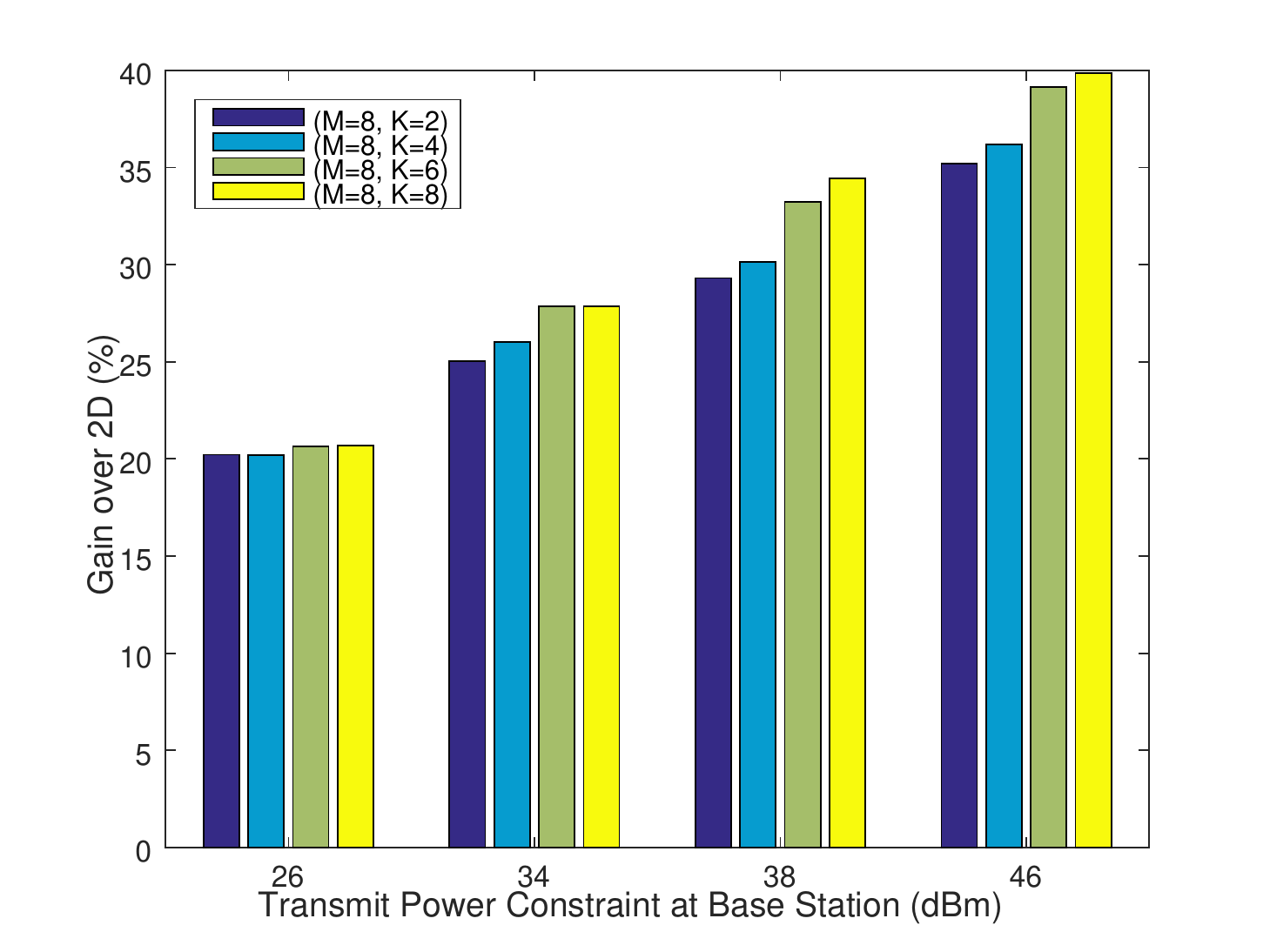}
  \caption{Gain of 3D beamforming over 2D beamforming expressed in percents for different values of transmit power available in BS.}
  \label{fig:wholecasesgain}
\end{figure}

\begin{figure}[h!]
  \includegraphics[width=\linewidth]{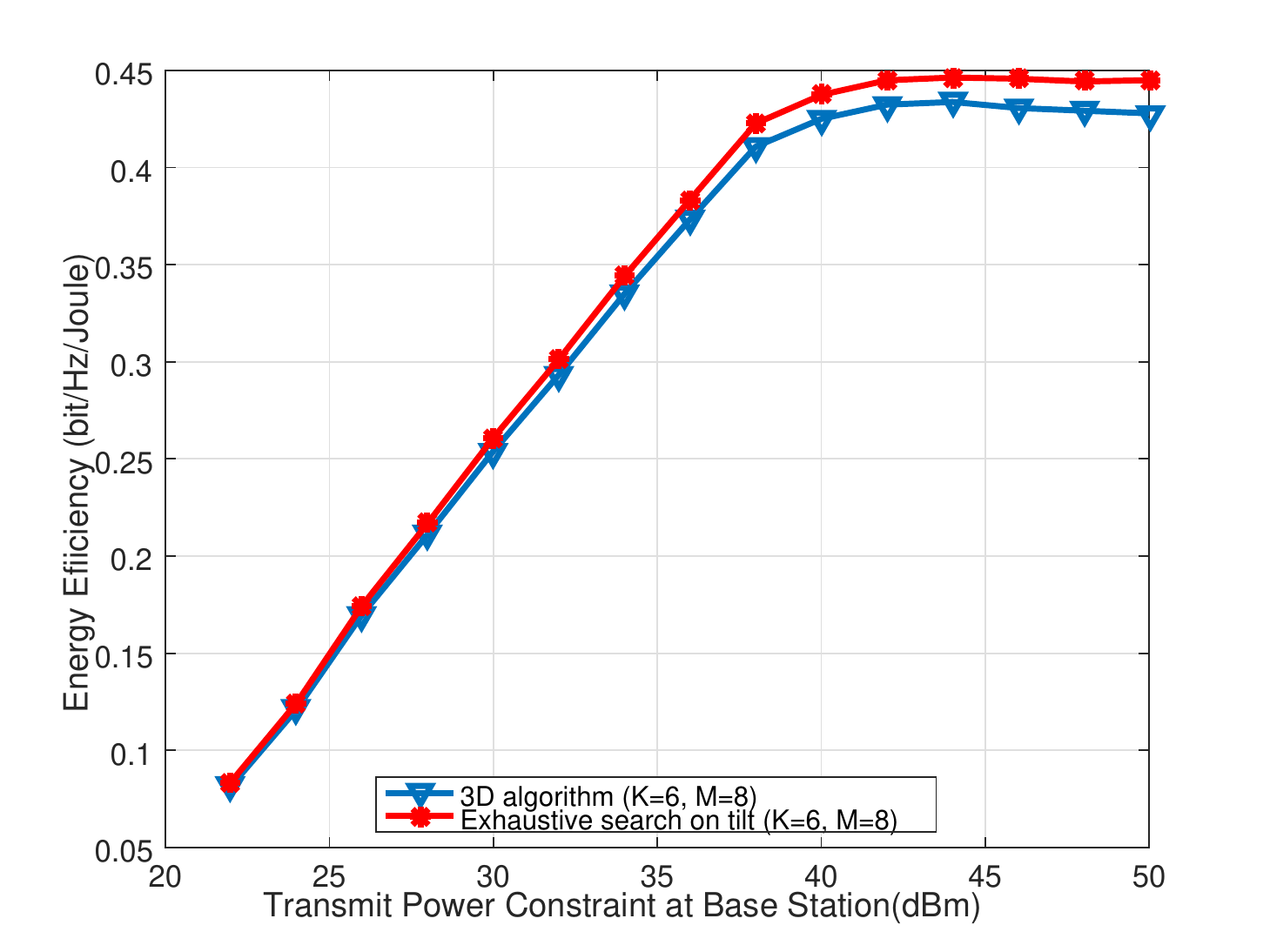}
  \caption{Performance comparison between clustering algorithm and exhaustive search over tilt angle of BS.}
  \label{fig:exhaustive}
\end{figure}

\begin{figure}[!h]
  \includegraphics[width=\linewidth]{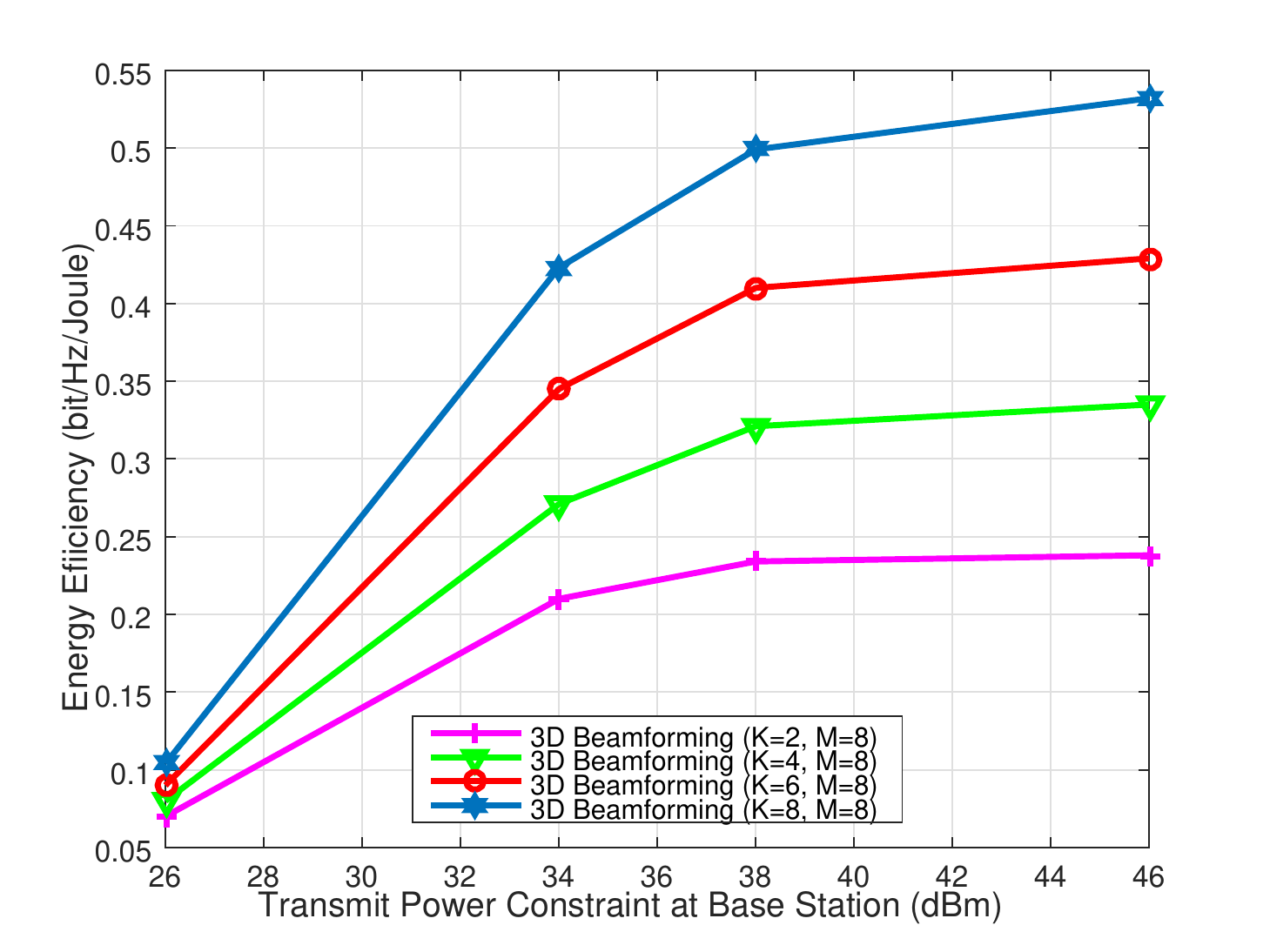}
  \caption{Impact of increasing the number of users in each cell.}
  \label{fig:userincrement}
\end{figure}

\section{Conclusion} \label{sec:CONCLUSION}
Downlink transmission in a cellular network with hexagonal shaped cells that  simultaneously employs beamforming and tilt angle adaptation of the BSs' antenna array has been investigated. First by employing fractional programming and introducing new variables, problem of maximizing total instantaneous EE with respect to beamforming vectors and tilt angle of BSs' of the network has been converted into multiple tractable problems which are distributed between BSs. Closed form solution of the optimum beamforming vectors have been found which are functions of channel gain of users, users' angle of arrival and tilt angle of the BS. Then finding optimum tilt angle has been investigated which is a non-convex problem and has multiple maximums. To overcome this difficulty, the optimum tilt angle has been found efficiently using a clustering algorithm. This algorithm classifies users based on their angle of arrival to the BS's antenna array and then aims the array's beam to a specific cluster that maximizes total instantaneous EE of the network. The clustering algorithm has significantly less complexity compared to exhaustive search over tilt angle but its performance is near the same as exhaustive search.

~\\~\\

\textbf{Soheil Khavari} received his B.Sc. from Amirkabir University of Technology (AUT), Tehran, Iran, in 2013 and his M.Sc. from Iran University of Science and Technology (IUST), Tehran, Iran in 2016, all in electrical engineering. His current research interests include wireless communication and signal processing.

~\\~\\

\textbf{S. Mohammad Razavizadeh} is an assistant professor in the school of Electrical Engineering at Iran University of Science and Technology (IUST), Tehran, Iran. His research interests are in the area of signal processing for wireless communications and cellular networks. He is a Senior Member of the IEEE.

\end{document}